\documentclass[preprint,12pt]{elsarticle}
\usepackage{amssymb}
\usepackage{amsthm}
\usepackage{graphicx}
\usepackage{graphics}
\usepackage{bm}   
\usepackage{color}
\usepackage{amsmath}
\usepackage{balance}
\usepackage{subfigure}
\usepackage{relsize}
\usepackage{float}
\newcommand{\mathsym}[1]{{}}
\newcommand{\unicode}[1]{{}}
\usepackage{setspace}
\DeclareMathOperator{\sech}{sech}
\newtheorem{thm}{Theorem}
\journal{Chaos, Solitons and Fractals}

\begin{document}

\begin{frontmatter}
\title{Higher order smooth positon and breather positon solutions of an extended nonlinear Schr\"{o}dinger equation with the cubic and quartic nonlinearity}

\author[inst1]{S. Monisha}

 \address[inst1]{Department of Nonlinear Dynamics, Bharathidasan University,\\Tiruchirappalli - 620024, Tamil Nadu, India.}

\author[inst2]{N. Vishnu Priya}
\author[inst1]{M. Senthilvelan}
\author[inst3]{S. Rajasekar}
\address[inst2]{Department of Mathematics, Indian Institute of Science,\\ Bangalore - 560012, Karnataka, India.}
  \address[inst3]{Department of Physics, Bharathidasan University,\\Tiruchirappalli - 620024, Tamil Nadu, India.}    

\begin{abstract}
%% Text of abstract
We construct certain higher order smooth positon and breather positon solutions of an extended nonlinear Schr\"{o}dinger equation with the cubic and quartic nonlinearity. We utilize the generalized Darboux transformation method to construct the aforementioned solutions. The three well-known equations, namely nonlinear Schr\"{o}dinger equation, Hirota equation, and generalized nonlinear Schr\"{o}dinger equation, are sub-cases of the considered extended nonlinear Schr\"{o}dinger equation. The solutions which we construct are more general. We analyze how the positon and breather positon solutions of the constituent equations get modified by the higher order nonlinear and dispersion terms. Our results show that the width and direction of the smooth positon and breather-positon solutions are highly sensitive to higher-order effects. Further, we carryout an asymptotic analysis to predict the behaviour of positons. We observe that during collision positons exhibit a time-dependent phase shift. We also present the exact expression of time-dependent phase shift of positons.  Finally, we show that this time-dependent phase shift is directly proportional to the higher order nonlinear and dispersion parameters. 
\end{abstract}

\begin{keyword}
Higher order nonlinear Schr\"{o}dinger equation  \sep positon solutions \sep Generalized Darboux transformation method
\end{keyword}

\end{frontmatter}

\section{Introduction}

Exploring the localized solutions of certain nonlinear integrable partial differential equations is an active field of research in optics, plasma physics, astrophysics, oceanography and so on \cite{Y1,Y2,Y3,Y4}. The solutions of nonlinear evolution equations, namely solitons, breathers, kinks, vortex solitons, dissipative solitons, oscillons, and rational solutions, model many real-world phenomena such as tsunami waves, rogue waves, tidal waves, cyclonic waves, waves in nonlinear optical fibres, shallow water waves and so on \cite{soliton,dissipative}. In this direction, a new type of solution called positons are constructed for the Korteweg-de Vries (KdV) equation \cite{M1,common}. Positons exhibit similar structures as that of solitons in the long range. The positon solution for the KdV equation was obtained by imposing a positive eigenvalue (spectral parameter) in the generalized Darboux transformation (GDT) method \cite{M1}. For the negative eigenvalues, one obtains soliton and negaton solutions \cite{boris,neg}. It has also been shown that positon solution can model shallow water rogue waves (extreme waves in oceanic circumstances) \cite{Matveev4}.
 \par  The positon solution of the KdV equation presented by Matveev had a spectral singularity \cite{Matveev3}. Recently, for the KdV equation, Cen \textit{et al}., have circumvented this singularity by relaxing the spectral parameter ($\lambda$) be complex  \cite{kdv,SineG}.  Subsequently the resultant solutions are named as smooth positons or degenerated soliton solutions in the literature \cite{jing,bk,ckdv}.  
 \par  It is well known that during collision, solitons display constant phase shift. Differing from this, the smooth positons exhibit time-dependent phase shift when they collide with each other \cite{Hirota2}.  Due to this time-dependent phase shift, the positons exhibit different behaviour in short and long times. Positons travel simultaneously like a single component at small time scales and then they separate from each other and produce equal amplitude one soliton constituents with equal energy at large times \cite{kdv2}. As far as multi-solitons are concerned they split into different amplitude one soliton constituents with different speeds and energy. We mention here that it is not possible to obtain positon solutions by imposing equal speed parameters in multi-soliton solutions. For small time scale, positons may model the well-known tidal bore phenomenon \cite{tidal}.  Tidal bore phenomenon is nothing but in the river, many equal highest amplitude waves travel together for long distances say more than hundred kilometers. Since the positons do not exhibit energy exchange during collision one can prevent data loss. Hence the study of positons will also be useful for optics community \cite{Hirota2,tidal}.
 \par The singular and smooth positon solutions have been constructed for a class of integrable equations \cite{SineG1,p1,sasa,Liu1,Liu2,song,p2,mKdv,kundu,cmkdv,chin,w1}.  Breather-positons (b-p) are equal amplitude breathers (localized periodic waves on constant background) and they travel with equal speed \cite{A1,A2,cmkdv2,Kundu2,sasa2,mb}.  The central region of the b-p solution exhibits a rogue wave like structure. The nonlinear Schr\"{o}dinger equation and its higher order generalizations have been considered in several fields including fluid dynamics, birefringent optical fiber, shallow water surfaces, Heisenberg spin systems and ocean waves \cite{b1,b2,b3,b4,fib,fib2,ocean1,ocean2}. In this work, we construct certain smooth positons and b-p solutions of an extended nonlinear Schr\"{o}dinger equation (ENLSE) \cite{int,bim,hrw,yulanma}  with the cubic and quartic nonlinearity. 
\par More specifically, we consider the following ENLSE \cite{Anki},
\begin{eqnarray}
	iq_t+q_{xx}+2|q|^2q - i\alpha(q_{xxx}+6 q_{x}|q|^2)+\gamma(q_{xxxx}+8|q|^2q_{xx}\notag \\+2q^2q^*_{xx}+4q|q_x|^2+6q^{2}_xq^*+6|q|^4q)=0,
	\label{eq1}
\end{eqnarray}
with $q$ represents the wave envelope with $x$ being the propagation variable in the moving frame with time $t$. The coefficients $\alpha$ and $\gamma$ are associated with the third order ($q_{xxx}$) and fourth order dispersions ($q_{xxxx}$) respectively.  Equation \eqref{eq1} reduces to the  standard normalized NLS equation when both the coefficients $\alpha$ and $\gamma$ are zero  \cite{nls,P1,P2,P3}. If $\alpha \neq 0$, $\gamma = 0$ we obtain the Hirota equation \cite{hirota,Anki2} and for $\gamma\neq0$, $\alpha = 0$ we get another well known generalized NLS equation \cite{lpd}. Using the GDT, we construct second, third and fourth order smooth positons with zero seed solution and some higher order b-p solutions with plane wave seed solution. To the best of our knowledge, the solutions constructed in this paper are all new. Another aim of this work is to analyze how the higher order nonlinear terms affect the basic smooth positon and breather positon solutions. We investigate this aspect by varying the higher order nonlinear and dispersion parameters, namely $\alpha$ and $\gamma$ and studying the outcome. Our investigations reveal that their variations produce large compression in the width of the waves. We also observe tilt and deviations in their orientations. While increasing the value of higher order nonlinear and dispersion parameters, the distance in-between the two waves decreases in higher order smooth positons and b-p solutions.  Further, we carryout an asymptotic analysis on the second order smooth positon solutions in order to predict the behaviour of it at small and large time scales.  We observe that in contrast to soliton solutions the smooth positons exhibit time dependent phase shift.  Due to this time dependency, the second order smooth positon moves as a single component for small time and they split into two one solitons at long time which we also demonstrate graphically. The phase shift also depends on the higher order nonlinear and dispersion parameters.  By increasing the values of these two parameters, $\alpha$ and $\gamma$, the phase shift also increases which we also demonstrate pictorially. 

\par The presentation is organized as follows. In Sec.~2, we derive $N^{th}$ order DT solution formula of Eq.~\eqref{eq1}. We construct higher order smooth positon and b-p solutions through GDT method in Secs.~3 and 4 respectively and we also analyze the consequences of higher order nonlinear terms by varying the value of the parameters $\alpha$ and $\gamma$. In Sec.~5, we carry out the asymptotic analysis on the second order smooth positon solution and determine its time dependent displacement. We present the outcome of our investigations in Sec.~6.
\section {Generalized Darboux Transformation of \eqref{eq1} }
Darboux transformation (DT) is one of the well-known methods which is used to derive various kinds of solutions of  integrable nonlinear partial differential equations \cite{Matveevbook}.  It keeps the form of eigenvalue equations unchanged and establishes a relation between old and new potentials after appropriate transformation.  Using DT method one can derive multi-soliton solutions, rogue wave solutions of different order and higher order breathers of integrable nonlinear evolution equations \cite{bq1,bq2,bq3,bq4,bq5,bq6,bq8}. To derive positons and negatons of KdV equation, Matveev introduced the generalized DT (GDT) method  \cite{M1}.  In the GDT method, basic solutions of eigenvalue equations are expanded at a single eigenvalue using Taylor expansions.  We can obtain multi-solitons, higher order breathers, positons, negatons and rogue wave solutions using GDT method.  To derive GDT of \eqref{eq1} we consider 
the Lax pair of it as
\begin{subequations}
\begin{eqnarray}
	\begin{aligned}
	\Psi_x = L \Psi,\\
	\Psi_t = G \Psi, \label{lp}
\end{aligned}
\end{eqnarray}
where 
\begin{eqnarray}
		L = \lambda J + U,\quad G= 2 \lambda^2 J + 2 \lambda U + B - \alpha R + \gamma K,  
\end{eqnarray}
\begin{eqnarray*}
\begin{aligned}	
U= 
\begin{pmatrix}
0 & ~q \\ \\-q^* &0	
\end{pmatrix},\qquad 
J = \begin{pmatrix}
-i &~0 \\\\ 0 &i
\end{pmatrix},\qquad  
B= \begin{pmatrix}
i |q|^2 & i q_x \\\\ i q_x ^* &- i |q|^2 
\end{pmatrix},
\end{aligned}
\end{eqnarray*}

\begin{eqnarray}
	\begin{aligned}
	M=	\begin{pmatrix}
			q q_x^* - q^*q_x & -(2|q|^2 q +q_{xx})\\ \\2|q|^2 q^* +q^*_{xx}& -(q q_x^* - q^* q_x)
		\end{pmatrix}, \qquad K = \begin{pmatrix}
		i K_1 & K_2 \\ \\- K_2^* & ~- i K_1 
	\end{pmatrix},
	\end{aligned}
\end{eqnarray}
with
\begin{eqnarray}
	R&=& 4 \lambda^3 J + 4 \lambda^2 U+ 2 \lambda B +M,\notag \\
	K_1& =& 3 |q|^4 - |q_x|^2 + qq_{xx}^* + q^* q_{xx} - 2 i \lambda( q^* q_x - q q^*_x) - 4 \lambda^2 |q|^2 + 8 \lambda^4,\notag \\
	K_2&=& 6 i |q|^2 q_x + i q_{xxx} + 2 \lambda q_{xx}+ 4 \lambda |q|^2 q - 4i \lambda^2 q_x - 8 \lambda ^3 q.
\end{eqnarray}
\end{subequations}
Here, $\Psi = ( f,g)^T$  is the vector eigenfunction and $\lambda$ is the spectral parameter. The compatability condition of (\ref{lp}), $L_t - G_x + [L,G]=0$, gives Eq.~\eqref{eq1}. 
The DT formula of ENLSE has already been constructed for the breather and rogue wave solutions in \cite{hbs}. Hence, we express only the explicit solution formula of Eq.~\eqref{eq1} here. Solving the Lax pair equations \eqref{lp} with $N$-eigenvalues $\lambda_i, i=1,2,3,...,N$, we can generate the eigenfunctions $\Psi_i, i=1,2,...,N$. Then, $Nth$ iteration of DT gives the solution formula of Eq.~\eqref{eq1} in the form
\begin{subequations}
\begin{eqnarray}
	q_{N}=q_0-2i\dfrac{|D_{1N}|}{|D_{2N}|},
	\label{qn}
\end{eqnarray}
where $q_0$ is the seed solution with
\begin{equation}
	D_{1N}=
		\begin{pmatrix}f_1&g_{1}&\lambda_{1}f_1&\lambda_{1}g_1&\lambda_1^2f_1&\lambda_{1}^2g_1&\cdots&\lambda_{1}^{N-1}f_1&\lambda_1^Nf_1\\-g_1^*&f_{1}^*&-\lambda_{1}^{*}g_1^*&\lambda_{1}^{*}f_1^*&-\lambda_1^{*2}g_1^*&\lambda_{1}^{*2}f_1^*&\cdots&-\lambda_{1}^{*N-1}g_1^*&-\lambda_{1}^{*N}g_1^*\\f_2&g_{2}&\lambda_{2}f_2&\lambda_{2}g_2&\lambda_2^2f_2&\lambda_{2}^2g_2&\cdots&\lambda_{2}^{N-1}f_2&\lambda_2^Nf_2\\
			\vdots&\vdots&\vdots&\vdots&\vdots&\vdots&\cdots&\vdots&\vdots\\f_N&g_{N}&\lambda_{N}f_N&\lambda_{N}g_N&\lambda_N^2f_N&\lambda_{N}^2g_N&\cdots&\lambda_{N}^{N-1}f_N&\lambda_N^Nf_N
	\end{pmatrix},
	\label{D1}
\end{equation}
and
\begin{equation}	
D_{2N}=
\begin{pmatrix}f_1&g_{1}&\lambda_{1}f_1&\lambda_{1}g_1&\lambda_1^2f_1&\lambda_{1}^2g_1&\cdots&\lambda_{1}^{N-1}f_1&\lambda_1^{N-1}g_1\\-g_1^*&f_{1}^*&-\lambda_{1}^{*}g_1^*&\lambda_{1}^{*}f_1^*&-\lambda_1^{*2}g_1^*&\lambda_{1}^{*2}f_1^*&\cdots&-\lambda_{1}^{*N-1}g_1^*&\lambda_{1}^{*N-1}f_1^*\\f_2&g_{2}&\lambda_{2}f_2&\lambda_{2}g_2&\lambda_2^2f_2&\lambda_{2}^2g_2&\cdots&\lambda_{2}^{N-1}f_2&\lambda_2^{N-1}g_2\\
	\vdots&\vdots&\vdots&\vdots&\vdots&\vdots&\cdots&\vdots&\vdots\\f_N&g_{N}&\lambda_{N}f_N&\lambda_{N}g_N&\lambda_N^2f_N&\lambda_{N}^2g_N&\cdots&\lambda_{N}^{N-1}f_N&\lambda_N^{N-1}g_N
\end{pmatrix}.
	\label{D2}
\end{equation}
\end{subequations}
\par In Eqs.~\eqref{D1} and \eqref{D2}, $f_i$'s and $g_i$'s are eigenfunctions corresponding to the eigenvalues $\lambda_i$, $i=1,2,3...$. Using the solution formula \eqref{qn} we can derive various kinds of solutions including soliton, breather and rational solutions of Eq.~\eqref{eq1} with appropriate seed solutions.  \par We can construct degenerate solutions of the ENLSE by limiting the spectral parameters as $\lambda_i = \lambda_1 +\epsilon, i=2,3....N,$ where $\epsilon$ is a small parameter. We construct the GDT of ENLSE by expanding the eigenfunctions in Taylor series at $\epsilon$. The determinant representation of the solution formula \eqref{qn} can be generalized and rewritten in the following form  to derive the degenerate solutions of the ENLSE, that is
\begin{subequations}\label{qN}
\begin{eqnarray}
	q_{[N]}=q_0-2i\dfrac{D_1[N]}{D_2[N]},
	\end{eqnarray}
where
\begin{eqnarray}
	D_k[N]=\left|\lim_{\epsilon\rightarrow0}\dfrac{\partial^{N_i-1}}{\partial\epsilon^{N_i-1}}(D_{kN}^{'})\right|_{2N\times2N},\label{c}
\end{eqnarray}
\end{subequations}
with $D_{kN}^{'}=(D_{kN})_{ij}(\lambda_1+\epsilon)$, $k=1,2$ and  $N_i=\left[\dfrac{i+1}{2}\right]$, where $[i]$ is the floor function. The explicit positon solutions of ENLSE can be obtained through the modified DT formula \eqref{qN}.
\begin{figure}
	\includegraphics[width=\linewidth]{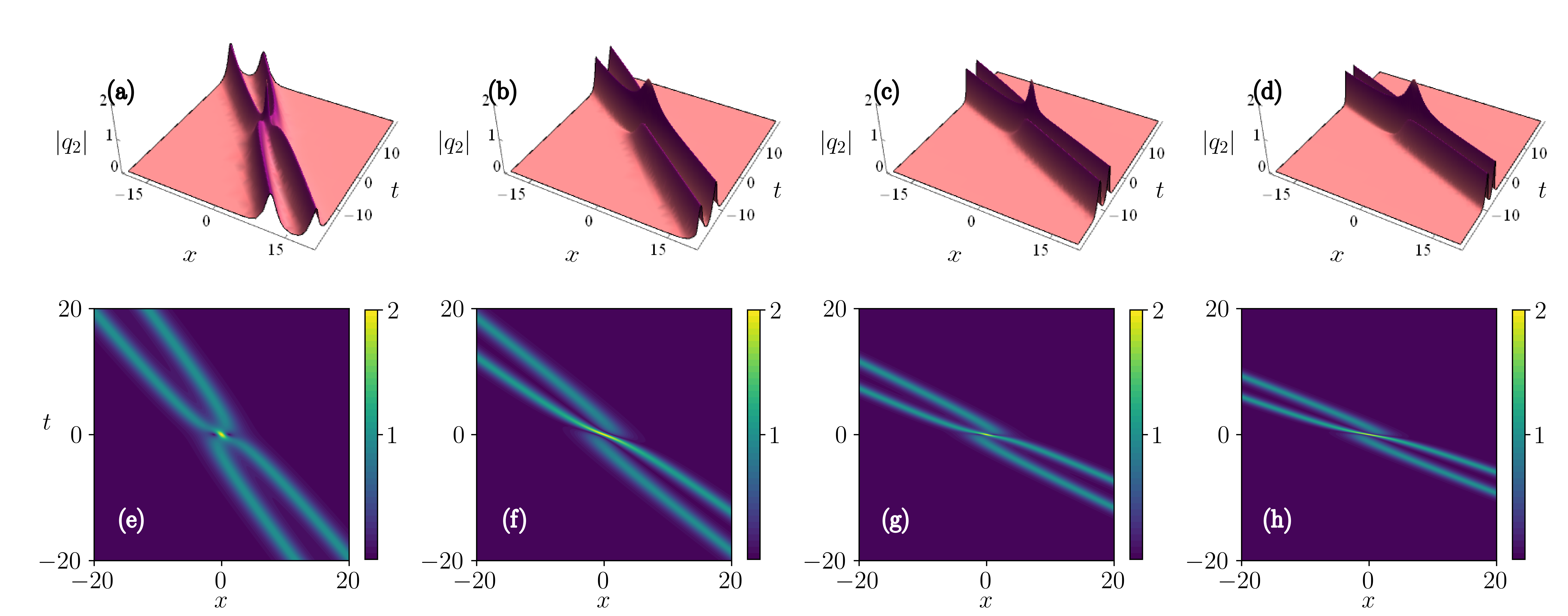}
	\caption{ Second order smooth-positon solution of ENLSE with the parameter values $\lambda_1=0.2+0.5i$, (a) $\alpha= \gamma=0$, (b) $\alpha=1;\;\gamma=0$, (c) $\alpha=0;\;\gamma=1$, (d) $\alpha= \gamma = 1$.  Figs.~(e)-(h) are corresponding contour illustration of Figs.~(a)-(d) respectively.} \label{sp22}
\end{figure}
\section{Smooth positon solutions of ENLSE}  
\subsection{Second order smooth positon solution}
\begin{thm}
 The explicit form of second order smooth positon solution of Eq.~\eqref{eq1} is 
\begin{subequations}\label{sp2}
	\begin{equation}
		q_{2}=\dfrac{\Omega_{12}}{\Omega_{22}},
	\end{equation}
	where
	\begin{eqnarray}
		\Omega_{12}&=&	4 (\lambda_1 -\lambda_1 ^*) 
		(e^{2 i x \lambda_1 +4 i t (\lambda_1 ^2-2 \alpha  \lambda_1 ^3-4 \gamma  \lambda_1 ^4+2 \lambda_1 ^{*2} (2 \alpha  \lambda_1 ^*+4 \gamma
			\lambda_1 ^{*2}-1))} (4 t \lambda_1  (3 \alpha  \lambda_1 \notag\\&&-1+8 \gamma  \lambda_1 ^2) (\lambda_1 -\lambda_1 ^*)+x (\lambda_1^*
		-\lambda_1 )-i)-e^{2 i \lambda_1 ^* (x+2 t \lambda_1 ^* (-1+2 \alpha  \lambda_1 ^*+4 \gamma  \lambda_1 ^{*2}))} \notag\\&&\times (i+x (\lambda_1 ^*-\lambda_1
		)+4 t (\lambda_1 -\lambda_1 ^*) \lambda_1 ^* (3 \alpha  \lambda_1 ^*+8 \gamma  \lambda_1 ^{*2}-1))),\notag\\
		\Omega_{22}&=&e^{4 i x \lambda_1 ^*}-
		2 e^{2 i x (\lambda_1 +\lambda_1 ^*)+4 i t (\lambda_1 ^2-2 \alpha  \lambda_1 ^3-4 \gamma  \lambda_1 ^4+\lambda_1 ^{*2}  (2 \alpha
			\lambda_1 ^*-1+4 \gamma  \lambda_1 ^{*2}))}  \times
		(2 x^2 (\lambda_1 -\lambda_1 ^*)^2\notag\\&&-1+32 t^2 \lambda_1  (\lambda_1 -\lambda_1 ^*)^2 \lambda_1 ^* (1+9 \alpha ^2 |\lambda_1|^2+64 \gamma ^2 |\lambda_1| ^4+3 \alpha  (\lambda_1 +\lambda_1 ^*) \notag\\&& \times(8 \gamma  |\lambda_1|^2-1)-8 \gamma
		(\lambda_1 ^2+\lambda_1 ^{*2}))-8 t x (\lambda_1 -\lambda_1 ^*)^2 (3 \alpha  \lambda_1 ^2+8 \gamma  \lambda_1 ^3-\lambda_1 \notag\\&&+\lambda_1
		^* (3 \alpha  \lambda_1 ^*+8 \gamma  \lambda_1 ^{*2}-1)))+e^{4 i x \lambda_1 +8 i t (\lambda_1 ^2-2 \alpha  \lambda_1 ^3-4 \gamma  \lambda_1 ^4+\lambda_1 ^{*2} (2 \alpha
			\lambda_1 ^*+4 \gamma  \lambda_1 ^{*2}-1))}.\notag\\
	\end{eqnarray}
\end{subequations}
\end{thm}
\begin{proof}
By setting the spectral parameter $\lambda = \lambda_{1}$ with vacumn seed solution $q_0=0$, we can show that the following eigenfunctions satisfy the Lax pair equations \eqref{lp}, that is 
\begin{subequations}
	\begin{eqnarray}
	f_1&=&exp(-i\lambda_1x-(2i \lambda_1 ^2 -4i\alpha \lambda_1^3 - 8 i \gamma \lambda_1^4)t), \\
	g_1&=&exp(i\lambda_1x+(2i \lambda_1 ^2 -4i\alpha \lambda_1^3 - 8 i \gamma \lambda_1^4)t).
	\label{ef}
\end{eqnarray} 
\end{subequations}
The choice $N=2$ in the solution formula \eqref{qN} provides
\begin{subequations} 
\begin{eqnarray}
	q_{2}=q_0-2i\dfrac{|D_{12}|}{|D_{22}|},
	\label{qsp2}
\end{eqnarray}
where
\begin{equation}
	D_{12}=\begin{pmatrix}f_1&g_1&\lambda_{1}f_1&\lambda_{1}^{2}f_1\\-g_1^*&f_1^{*}&-\lambda_1^*g_1^*&-\lambda_1^{*2}g_1^*\\f_2&g_2&\lambda_2f_2&\lambda_{2}^2f_2\\-g_2^*&f_{2}^*&-\lambda_2^*g_2^*&-\lambda_2^{*2}g_{2}^*\end{pmatrix},\;\;
D_{22}=\begin{pmatrix}f_1&g_1&\lambda_{1}f_1&\lambda_{1}g_1\\-g_1^*&f_1^{*}&-\lambda_1^*g_1^*&\lambda_1^{*}f_1^*\\f_2&g_2&\lambda_2f_2&\lambda_{2}g_2\\-g_2^*&f_{2}^*&-\lambda_2^*g_2^*&\lambda_2^{*}f_{2}^*\end{pmatrix}.
	\label{spD}
\end{equation}
\end{subequations}
\par Imposing the limit $\lambda_2 =\lambda_1 + \epsilon$ in Eqs.~\eqref{qsp2} and \eqref{spD} and expanding the eigenfunctions $f_1$ and $g_1$ in Taylor series at $\epsilon$, we obtain the following expressions for $D_{12}$ and $D_{22}$, namely
\begin{subequations}
\begin{eqnarray}
	\lim_{\epsilon\rightarrow 0}D_{12}&=&\begin{pmatrix}\vspace{0.1cm}f_1&g_1&\lambda_{1}f_1&\lambda_{1}^{2}f_1\\\vspace{0.15cm}-g_1^*&f_1^{*}&-\lambda_1^*g_1^*&-\lambda_1^{*2}g_1^*\\\vspace{0.1cm}f_{1l}'&g_{1l}'&((\lambda_1 + \epsilon)f_{1l})'&((\lambda_1 + \epsilon)^2f_{1l})'\\\vspace{0.1cm}-g_{1l}^{*'}&f_{{1l}}^{'*}&-({(\lambda_1^* + \epsilon)g_{1l}^*})'&-({(\lambda_1^* + \epsilon)^{2}g_{{1l}}^*})'\end{pmatrix},
\end{eqnarray}
\begin{eqnarray}
\lim_{\epsilon\rightarrow 0}	D_{22}&=&\begin{pmatrix}\vspace{0.1cm}f_1&g_1&\lambda_{1}f_1&\lambda_{1}g_1\\ \vspace{0.1cm}-g_1^*&f_1^{*}&-\lambda_1^*g_1^*&\lambda_1^{*}f_1^*\\\vspace{0.1cm}f_{1l}'&g_{1l}'&((\lambda_1 + \epsilon)f_{1l})'&((\lambda_1 + \epsilon)g_{1l})'\\\vspace{0.1cm}-g_{1l}^{*'}&{f_{{1l}}^*}'&-((\lambda_1^* + \epsilon)g_{1l}^*)'&((\lambda_1^* + \epsilon)f_{{1l}}^*)'\end{pmatrix},
	\label{spD1}
\end{eqnarray}
with
\begin{eqnarray}
f_{1l}&=&e^{-i(\lambda_1+\epsilon)x-(2i (\lambda_1+ \epsilon) ^2 -4i\alpha (\lambda_1+\epsilon)^3 - 8 i \gamma (\lambda_1+\epsilon)^4)t},\notag\\
g_{1l}&=&e^{i(\lambda_1+\epsilon)x+(2i (\lambda_1+\epsilon) ^2 -4i\alpha (\lambda_1+\epsilon)^3 - 8 i \gamma (\lambda_1+\epsilon)^4)t},
\end{eqnarray}
\label{mo}
\end{subequations}
 and prime in Eq.~\eqref{mo} represents differentiation with respect to $\epsilon$.  Substituting the above expressions (\ref{mo}) in (\ref{qsp2}) we end up (\ref{sp2}).
 \end{proof}
\par The second order smooth positon solution \eqref{sp2} and its corresponding contour plots are produced in Fig.~\ref{sp22}. We investigate the effect of higher order odd and even nonlinear terms on the solution by varying the value of the parameters $\alpha$ and $\gamma$. In Figure 1(a) we draw the second order smooth positon solution of the NLS equation ($\alpha= \gamma= 0$). If we consider only the third order dispersion term and neglect the fourth order term, that is $\gamma =0 $ and $\alpha = 1$, we arrive at the Hirota equation. The associated smooth positon solution is drawn in Fig.~1(b). Here, we can observe a shrink in the width of two smooth positons accompanied by a decrease in the distance between them. As far as smooth positons are concerned we observe a slight shift towards right in their orientation.
\par Now let us consider only the fourth order dispersion term in the picture, that is $\alpha=0$ and $\gamma=1$. The underlying equation becomes the fourth order NLS equation. In this case, we can see a higher compression in the width of smooth positons and a drastic change in their orientations, see Fig.~1(c). It is clear that the fourth order nonlinear term ($\gamma$) produces more compression and directional changes in the smooth positons when compare to the third order nonlinear term ($\alpha$). To observe the combined effect of both third and fourth order nonlinear terms, we set $\alpha= \gamma =1 $ and analyze how these two  higher order nonlinear and dispersion terms ($\alpha$ and $\gamma$) together effect the basic NLS positon solution. It is evident from Fig.~1(d) that these two higher order terms introduce a larger compression effect and appreciable changes in their orientations.
\begin{figure}[hb!]
	\includegraphics[width=\linewidth]{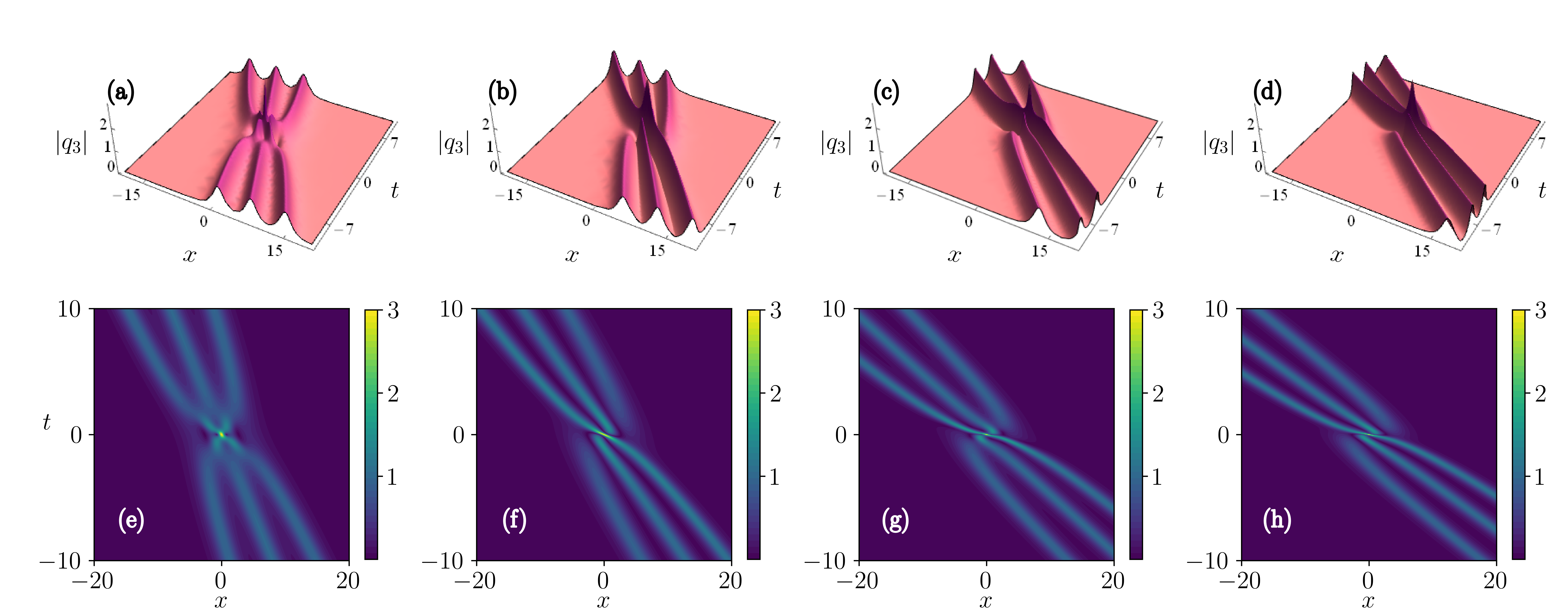}
	\caption{ Third order smooth-positon solution of ENLSE with the parameter values  $\lambda_1=0.2+0.5i$, (a) $\alpha= \gamma=0$, (b) $\alpha=1;\;\gamma=0$, (c) $\alpha=0;\;\gamma=1$, (d) $\alpha= \gamma = 1$.  Figs.~(e)-(h) are the corresponding contour illustration of Figs.~(a)-(d) respectively.}
\end{figure}

\subsection{Higher order smooth positon solutions}

\par To derive the third order smooth positon solution of ENLSE, we consider $N=3$ in Eq.~\eqref{qn} and impose a limit on the spectral parameters as $\lambda_i= \lambda_1+\epsilon, i = 2,3$. The resultant determinant takes the form
{\scriptsize \begin{subequations}\label{spD3}
		 \begin{equation}
		\lim_{\epsilon\rightarrow 0}D_{13}=\begin{pmatrix}\vspace{0.2cm}f_1&g_1&\lambda_1f_1&\lambda_1g_1&\lambda_1^2f_1&\lambda_{1}^{3}f_1\\\vspace{0.2cm}
		-g_1^*&f_1^{*}&-\lambda_{1}^*g_1^*&\lambda_{1}^*f_1^*&-\lambda_{1}^{*2}g_1^*&-\lambda_{1}^{*3}g_1^*\\\vspace{0.2cm}
		f_{1l}'&g_{1l}'&((\lambda_1 + \epsilon)f_{1l})'&((\lambda_1 + \epsilon)g_{1l})'&((\lambda_1 + \epsilon)^2f_{1l})'&((\lambda_1 + \epsilon)^3f_{1l})'\\\vspace{0.2cm}
		{-g_{1l}^*}'&{f_{1l}^*}'&-((\lambda_1^* + \epsilon)g_{1l}^*)'&((\lambda_1^* + \epsilon)f_{{1l}}^{*})'&-((\lambda_1^* + \epsilon)^{2}g_{1l}^*)'&-((\lambda_1^* + \epsilon)^{3}g_{1l}^*)'\\\vspace{0.2cm}
		f_{1l}''&g_{1l}''&((\lambda_1 + \epsilon)f_{1l})''&((\lambda_1 + \epsilon)g_{1l})''&((\lambda_1 + \epsilon)^2f_{1l})''&((\lambda_1 + \epsilon)^{3}f_{1l})''\\\vspace{0.2cm}
		{-g_{1l}^*}''&{f_{1l}^{*}}''&-((\lambda_1^* + \epsilon)g_{1l}^{*})''&((\lambda_1^* + \epsilon)f_{1l}^{*})''&-((\lambda_1^* + \epsilon)^{2}g_{1l}^{*})''&-((\lambda_1^* + \epsilon)^{3}g_{1l}^{*})''\end{pmatrix},	
\end{equation}
\begin{equation}	\lim_{\epsilon\rightarrow 0}D_{23}=\begin{pmatrix}\vspace{0.2cm}f_1&g_1&\lambda_1f_1&\lambda_1g_1&\lambda_1^2f_1&\lambda_{1}^{2}g_1\\\vspace{0.2cm}
		-g_1^*&f_1^{*}&-\lambda_{1}^*g_1^*&\lambda_{1}^*f_1^*&-\lambda_{1}^{*2}g_1^*&\lambda_{1}^{*2}f_1^*\\\vspace{0.2cm}
		f_{1l}'&g_{1l}'&((\lambda_1 + \epsilon)f_{1l})'&((\lambda_1 + \epsilon)g_{1l})'&((\lambda_1 + \epsilon)^2f_{1l})'&((\lambda_1 + \epsilon)^2g_{1l})'\\\vspace{0.2cm}
		{-g_{1l}^*}'&{f_{1l}^*}'&-((\lambda_1^* + \epsilon)g_{1l}^*)'&((\lambda_1^* + \epsilon)f_{2}^{*})'&-((\lambda_1^* + \epsilon)^{2}g_{1l}^*)'&((\lambda_1^* + \epsilon)^{2}f_{1l}^*)'\\\vspace{0.2cm}
		f_{1l}''&g_{1l}''&((\lambda_1 + \epsilon)f_{1l})''&((\lambda_1 + \epsilon)g_{1l})''&((\lambda_1 + \epsilon)^2f_{1l})''&((\lambda_1 + \epsilon)^{2}g_{1l})''\\\vspace{0.2cm}
		{	-g_{1l}^{*}}''&{f_{1l}^{*}}''&-((\lambda_1^* + \epsilon)g_{1l}^{*})''&((\lambda_1^* + \epsilon)f_{1l}^{*})''&-((\lambda_1^* + \epsilon)^{2}g_{1l}^
		{*})''&((\lambda_1^* + \epsilon)^{2}f_{1l}^{*})''\end{pmatrix}.	 	
\end{equation}\end{subequations}}
 \par  The third order smooth positon solution of Eq.~(\ref{eq1}) can be obtained by substituting the  expressions \eqref{spD3} in \eqref{qN}. The explicit form of the obtained solution is very lengthy and so we do not reproduce its explicit form here and present only the plot of third order smooth positon solution in Fig.~2. One can directly identify the third order smooth positon solution of the NLS equation from the constructed solution by setting $\alpha= \gamma =0$ in it. Restricting the parameters in the fashion $\alpha=1, \gamma=0$ and $\alpha =0, \gamma =1$, in the derived solution, we can visualize the third order smooth positon of the Hirota equation (Fig.~2(b)) and the fourth order NLS equation (Fig.~2(c)) respectively. To learn the higher order effect, we fix $\alpha = \gamma =1 $ in the constructed solution. When we increase the value of third order dispersion term $\alpha$, the position of the smooth positon changes, the distance between the positons decreases and the width of the positons get reduced.  By increasing the value of the parameter $\gamma$, one can visualize more compression and also appreciable changes in their orientation. The net effect observed in this case is higher than the one produced solely by the nonlinear term $\alpha$. When we increase the value of both the parameters $\alpha$ and $\gamma$ the positons move close to each other which is depicted in Fig.~2(d). Considering $N=4$ in Eq.~\eqref{qn} and imposing the spectral parameters be $\lambda_i= \lambda_1+\epsilon, i = 2,3,4$, we can derive the fourth order smooth positon solution of Eq.\eqref{eq1}. We have also analyzed the effect of higher order nonlinear and dispersion terms with the help of this solution as we have done in the second and third order smooth positon cases. Our investigations reveal that the fourth order positon solution also exhibits a similar behaviour as that of lower order positon solutions which can be confirmed from Fig.~3.  
 \begin{figure}
 	\includegraphics[width=\linewidth]{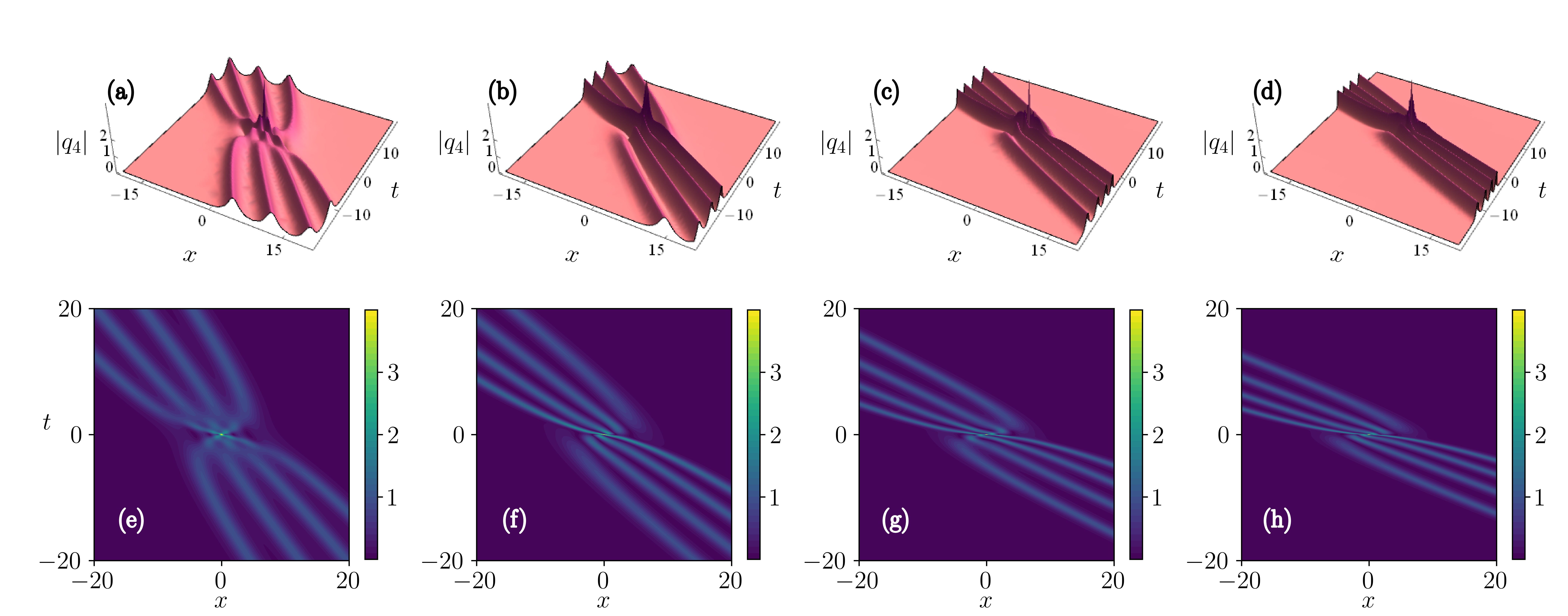}
 	\caption{Fourth order smooth-positon solution of ENLSE with parameter values $\lambda_1=0.2+0.5i$, (a) $\alpha= \gamma=0$, (b) $\alpha=1;\;\gamma=0$, (c) $\alpha=0;\;\gamma=1$, (d) $\alpha= \gamma = 1$.  Figs.~(e)-(h) are the corresponding contour illustration of Figs.~(a)-(d) respectively.}
 \end{figure}
 
\section{Breather positon solutions}
In Sec.~3, we discussed about the higher order smooth positon solutions of Eq.~\eqref{eq1}. Now, we construct the b-p solutions of the ENLSE.
\subsection{Second order b-p solution of ENLSE}
Upon solving the Lax pair  Eq.~\eqref{lp} with the plane wave solution $ q_0=  c e^{2ic^2(1+3 \gamma c^2)t}$ as seed solution, where $c$ is an arbitrary real parameter,  we obtain the following eigenfunctions, namely
 \begin{subequations}
\begin{eqnarray}
	f_1 = A_1 e^{ic^2(1+3 \gamma c^2)t}; \qquad g_1 = A_2 e^{-ic^2(1+3 \gamma c^2)t},\label{qbp}
\end{eqnarray}
where
\begin{eqnarray}
	A_1& =& \dfrac{c M_1}{i(\sqrt{\lambda_1^2 + c^2}+\lambda_1)} e^{\eta}+ \dfrac{c M_2}{i(-\sqrt{\lambda_1^2 + c^2}+\lambda_1)}e^{-\eta}, \nonumber\\
	\nonumber\\
	A_2 &=&M_1 e^{\eta}+ M_2 e^{-\eta}, \label{eta}
\end{eqnarray}
with
\begin{eqnarray}
	M_1& =&\dfrac{ i c + 2(\sqrt{\lambda_1^2 + c^2}+\lambda_1) }{2\sqrt{\lambda_1^2 + c^2}}, \qquad 	M_2 =\dfrac{- i c + 2(\sqrt{\lambda_1^2 + c^2}-\lambda_1) }{2\sqrt{\lambda_1^2 + c^2}}, \nonumber \\\nonumber\\
	\eta&= &i\sqrt{\lambda_1^2 + c^2}(x+2(\lambda_1 - 2 \alpha \lambda_1^2 - 
	4 \gamma \lambda_1^3 + c^2 (\alpha + 2 \gamma \lambda_1))t)\label{a4}.
\end{eqnarray}
\end{subequations}
\begin{figure}
	\includegraphics[width=\linewidth]{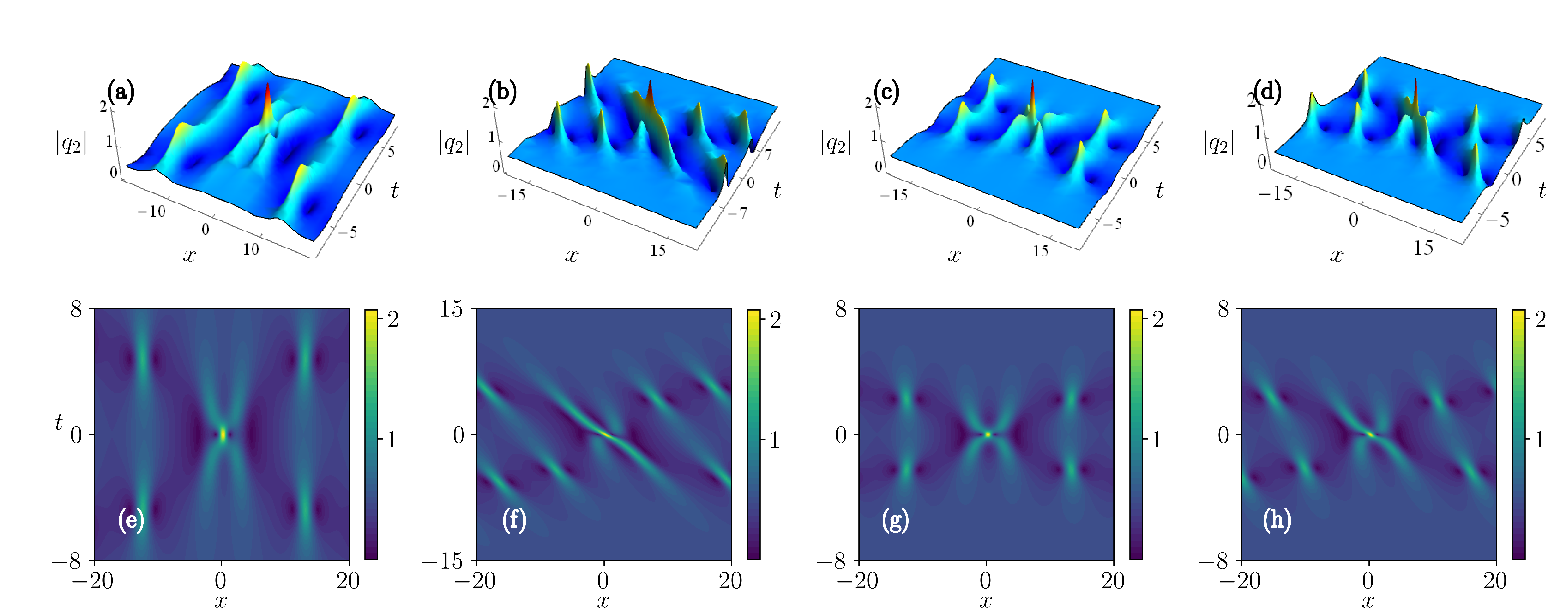}
	\caption{Second order b-p solution of ENLSE with the parameter values $\lambda_1=0.4i$ and $c=0.5$, (a) $\alpha= \gamma=0$, (b) $\alpha=1;\;\gamma=0$, (c) $\alpha=0;\;\gamma=1$, (d) $\alpha= \gamma = 1$.  Figs.~(e)-(h) are the corresponding contour illustration of Figs.~(a)-(d) respectively.}
\end{figure}
 \begin{figure}
	\includegraphics[width=\linewidth]{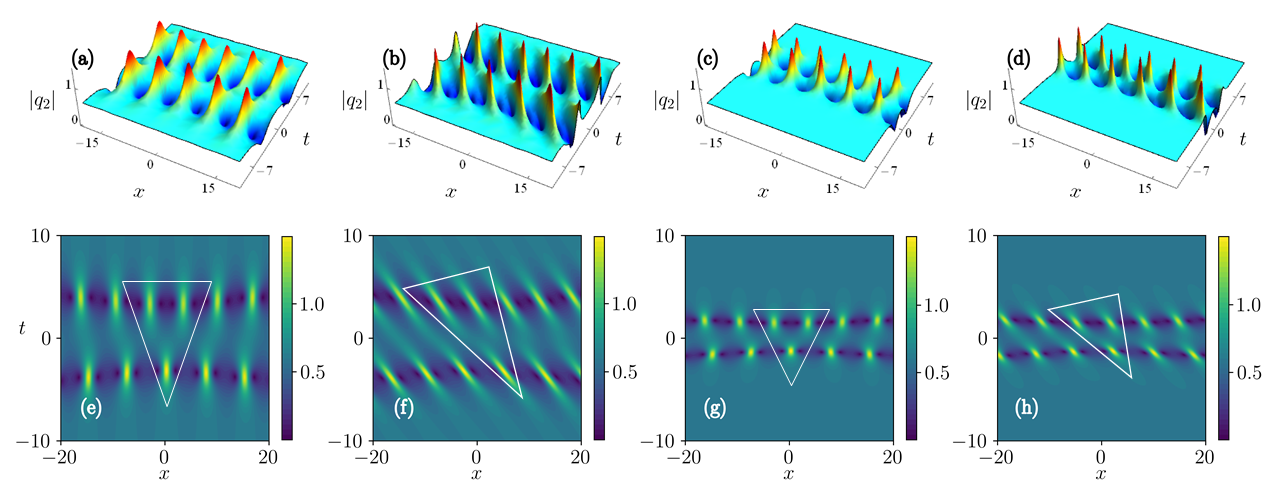}
	\caption{Second order triplet b-p solution of ENLSE with the parameter values $\lambda_1=0.4i$ and $c=0.6$, (a) $\alpha= \gamma=0$, (b) $\alpha=1;\;\gamma=0$, (c) $\alpha=0;\;\gamma=1$, (d) $\alpha= \gamma = 1$.  Figs.~(e)-(h) are the corresponding contour illustration of Figs.~(a)-(d).}
\end{figure}
\par The second order b-p solution of  the ENLSE can be identified by imposing the same limit on the spectral parameter $\lambda_2$, that is $\lambda_2 \rightarrow \lambda_{1}+ \epsilon$. 
Substituting the eigenfunctions found \eqref{qbp} in Eq.~\eqref{qsp2}, we obtain 
\begin{equation}
q_{2}^{[b-p]}=c e^{2ic^2(1+3 \gamma c^2)t}-2i\dfrac{|D_{12}^{[b-p]}|}{|D_{22}^{[b-p]}|}.\label{qbp2}
\end{equation} 
\par The 3-D plot of the second order b-p solution is shown in Fig.~4 (the explicit expression of Eq.~\eqref{qbp2} is very lengthy and so we do not print them here). In Fig.~4(a), we depict the second order b-p solution of the NLS equation ($\alpha= \gamma=0 $). In this figure, one may observe that the central region resembles the pattern of second order rogue wave. To examine the effect of higher order nonlinear and dispersion terms, first we choose the parameters in the fashion $\alpha=1$, $\gamma=0$ whose outcome is produced in Fig.~4(b). We observe a compression in the width of b-p and a change in their  orientation. Further, as we visualize the second order b-p gets tilted and the distance between two b-p also get decreased. As far as the fourth order NLS equation is concerned, that is $\alpha=0$, $\gamma=1$, the b-p solution exhibits larger compression and deviations in their orientation when compared to the previous case, see Fig.~4(c). For $\alpha= \gamma =1 $, the width of the b-p including the central region is highly compressed and they come very near to each other. The drastic changes in their orientation and tilt in the pulses by the nonlinear terms on the second order b-p solution are shown in Fig.~4(d). 
\par The dynamical changes in the structure of waves can be observed by adding an arbitrary constant $s_0 \epsilon$, that is $s_0 = s_{0r}+is_{0i}$ in the exponential function $\eta$ given in Eq.~\eqref{eta}. By choosing $s_{0r} =25$ and $s_{0i} =0$, the rogue wave like structure in the central region forms a triangular pattern of three single breather pulses (Fig.~5(a)). Analogously, we observe that the higher order nonlinear and dispersion terms exhibit  the same characteristics that we observed in the second order b-p case. The investigations reveal that the higher order parameters produce a compression in the width of smooth positons, appreciable changes in their direction and they also tilt the pulses which can all be visualized from Fig.~5.
 \begin{figure}
	\includegraphics[width=\linewidth]{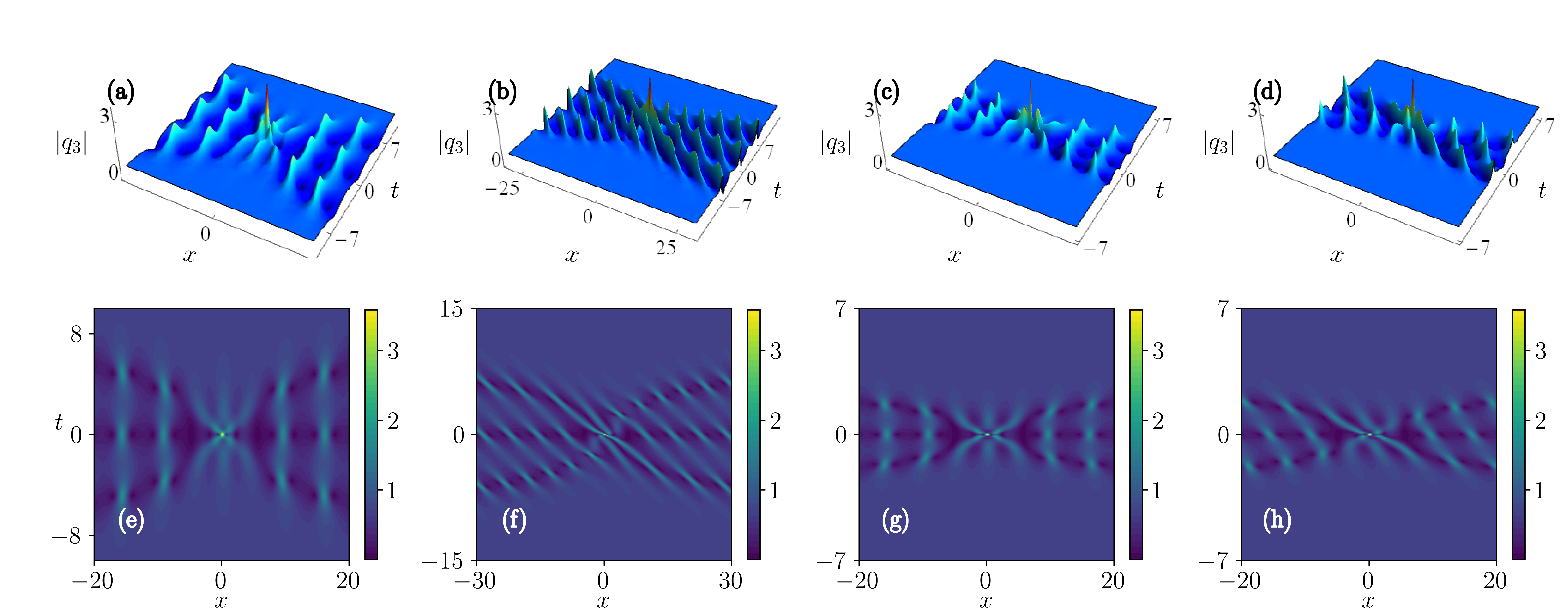}
	\caption{Third order b-p solution of ENLSE with the parameter values $\lambda_1=0.5i$ and $c=0.7$, (a) $\alpha= \gamma=0$, (b) $\alpha=1;\;\gamma=0$, (c) $\alpha=0;\;\gamma=1$, (d) $\alpha= \gamma = 1$. Figs.(e)-(h) are the corresponding contour illustration of Figs.(a)-(d).}
\end{figure}
\begin{figure}
	\includegraphics[width=\linewidth]{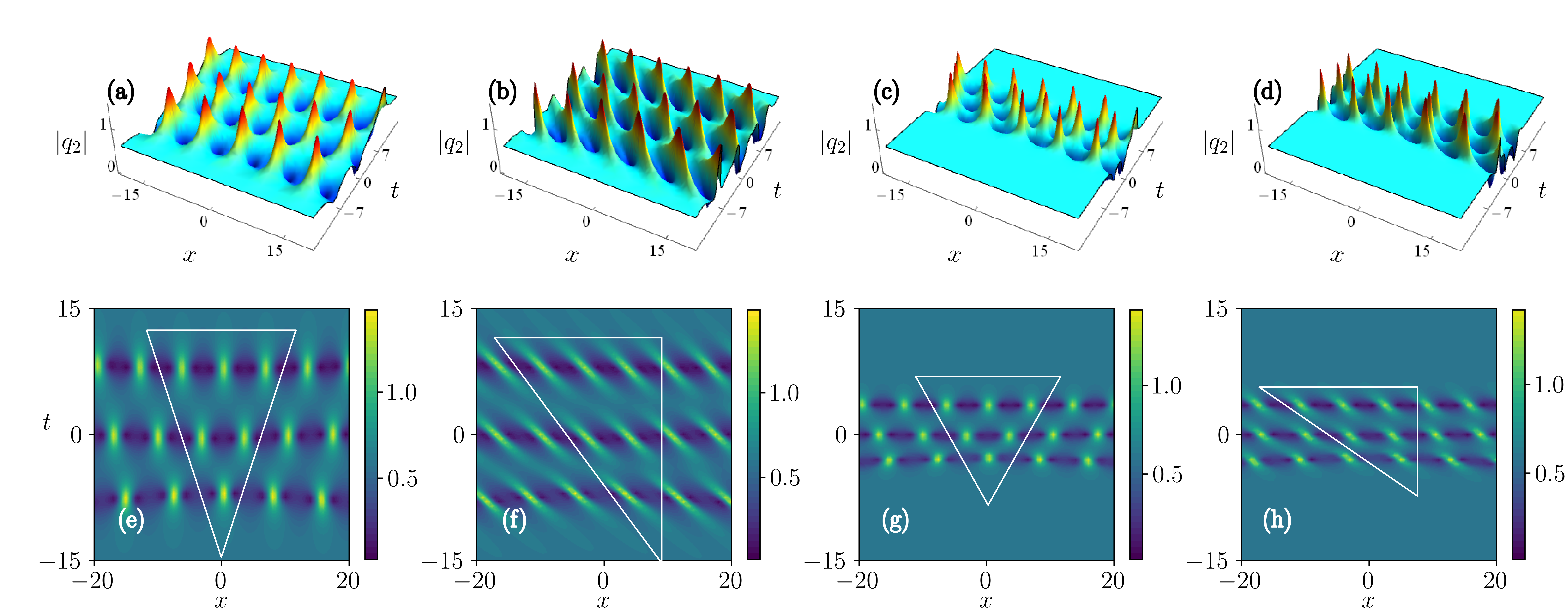}
	\caption{Third order triplet b-p solution of ENLSE with the parameter values $\lambda_1=0.4i$ and $c=0.7$ , (a) $\alpha= \gamma=0$, (b) $\alpha=1;\;\gamma=0$, (c) $\alpha=0;\;\gamma=1$, (d) $\alpha= \gamma = 1$. Figs.(e)-(h) are the corresponding contour illustration of Figs.(a)-(d).}
\end{figure}

\begin{figure}
	\includegraphics[width=\linewidth]{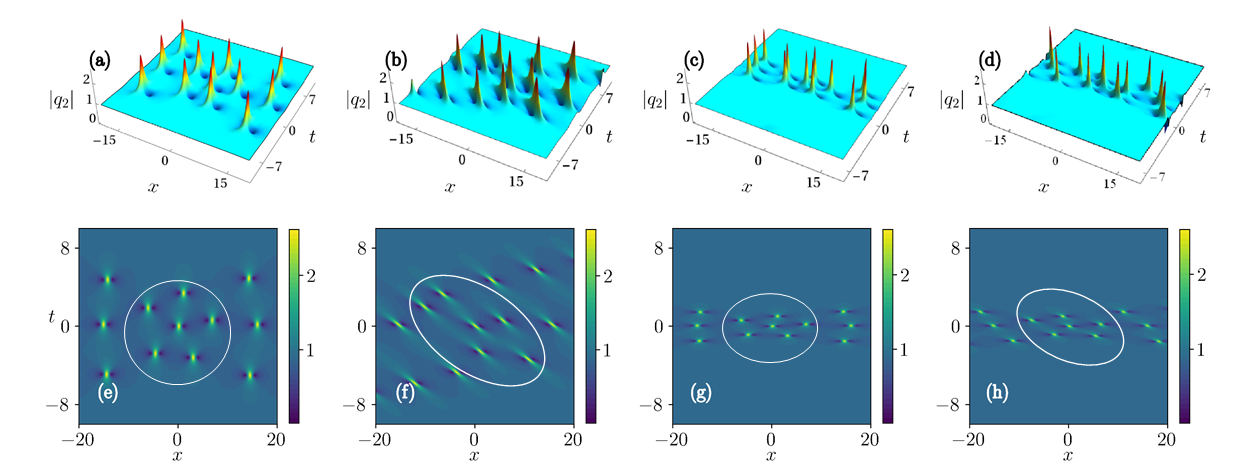}
	\caption{Third order circular b-p solution of ENLSE with the parameters $\lambda_1=0.85i$ and $c=0.9$ , (a) $\alpha= \gamma=0$, (b) $\alpha=1;\;\gamma=0$, (c) $\alpha=0;\;\gamma=1$, (d) $\alpha= \gamma = 1$.  Figs.~(e)-(h) are the corresponding contour illustration of Figs.~(a)-(d).}
\end{figure}

\subsection{ Higher order b-p solutions}
By considering $N=3$ in Eq.~\eqref{qn} with the eigenfunctions given in \eqref{qbp} and restricting the spectral parameters in the same manner as $\lambda_i \rightarrow \lambda_1 + \epsilon, \;i=2,3$, we can derive the third order b-p positon solution of the ENLSE. Since the explicit expression is quite cumbersome we present only the plots of the third order b-p solution in Fig.~6. The third order RW structure can be seen in the central region of the third order b-p solution. From Fig.~6, one may observe that the higher order nonlinear and dispersion terms produce the same effect on third order b-p solution which we come across earlier in the second order b-p case. To analyze the changes in the structure of b-p solution, we introduce two arbitrary constants, say $s_0 \epsilon$ and $s_1{\epsilon}$ in Eq.~\eqref{a4}, that is 
\(\eta= i\sqrt{\lambda_1^2 + c^2}(x+2(\lambda_1 - 2 \alpha \lambda_1^2 - 
4 \gamma \lambda_1^3 + c^2 (\alpha + 2 \gamma \lambda_1))t+s_0 \epsilon + s_1\epsilon^2)\) with $s_1 = s_{1r}+ i s_{1i}$. Upon introducing this, the central region splits up into six separate b-p pulses and forms a triangular pattern as shown in Fig.~7 for the values $s_{0r} = s_{1r}= 0$ and $s_{1r}=25,\; s_{1i} = 0$. If we consider $s_{1r}= s_{1i} = 600$, the central region forms a circular structure which is illustrated in Fig.~8. Higher compression effect and a change in their orientation is observed in the central region of the triangular and circular patterns for larger values of $\alpha$ and $\gamma$. These outcomes are demonstrated in Figs.~7 and 8. The arbitrary nonlinear parameter $\alpha$  is the main cause to produce the tilt in the b-p solution. Finally, we plot the fourth order b-p solution ($N=4$) in Fig. 9 and here also we observe the same behavioural changes. From the above outcome, we conclude that the width and direction of smooth positon and b-p solutions are highly sensitive to higher order effects. 
\begin{figure}
	\includegraphics[width=\linewidth]{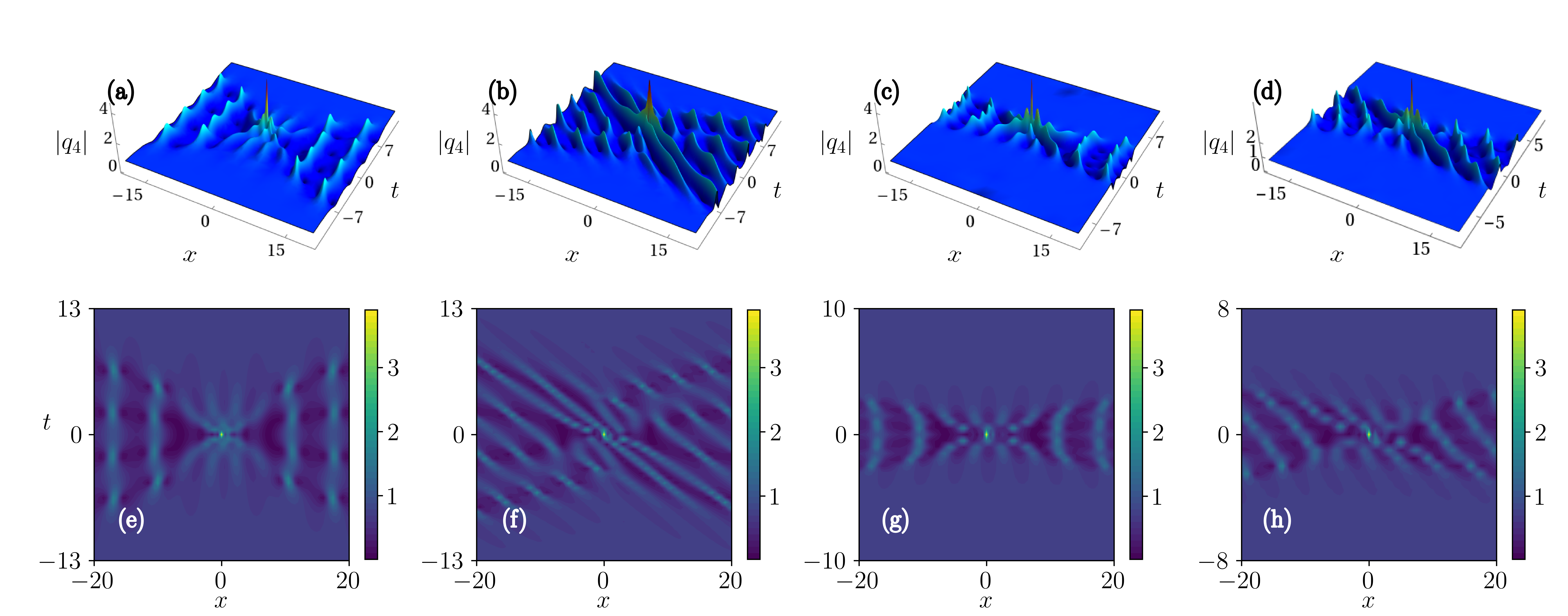}
	\caption{Fourth order b-p solution of ENLSE with the paramter values $\lambda_1=0.5i$ and $c=0.7$ , (a) $\alpha= \gamma=0$, (b) $\alpha=1;\;\gamma=0$, (c) $\alpha=0;\;\gamma=1$, (d) $\alpha= \gamma = 1$.Figs.~(e)-(h) are the corresponding contour illustration of Figs.~(a)-(d).}
\end{figure}
\begin{figure}
	\includegraphics[width=\linewidth]{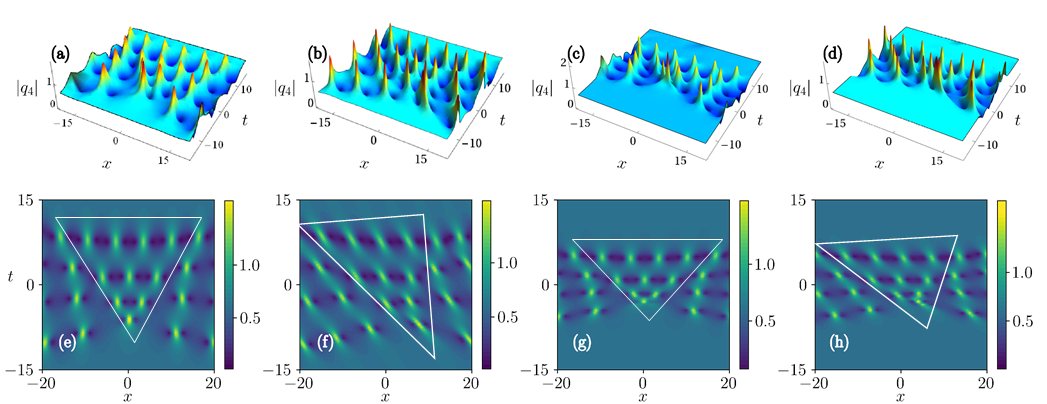}
	\caption{Fourth order triangular b-p solution of ENLSE with the paramter values $\lambda_1=0.5i$ and $c=0.7$ , (a) $\alpha= \gamma=0$, (b) $\alpha=0.5;\;\gamma=0$, (c) $\alpha=0;\;\gamma=0.5$, (d) $\alpha= \gamma = 0.5$. Figs.~(e)-(h) are the corresponding contour illustration of Figs.~(a)-(d).}
\end{figure}
\begin{figure}
	\includegraphics[width=\linewidth]{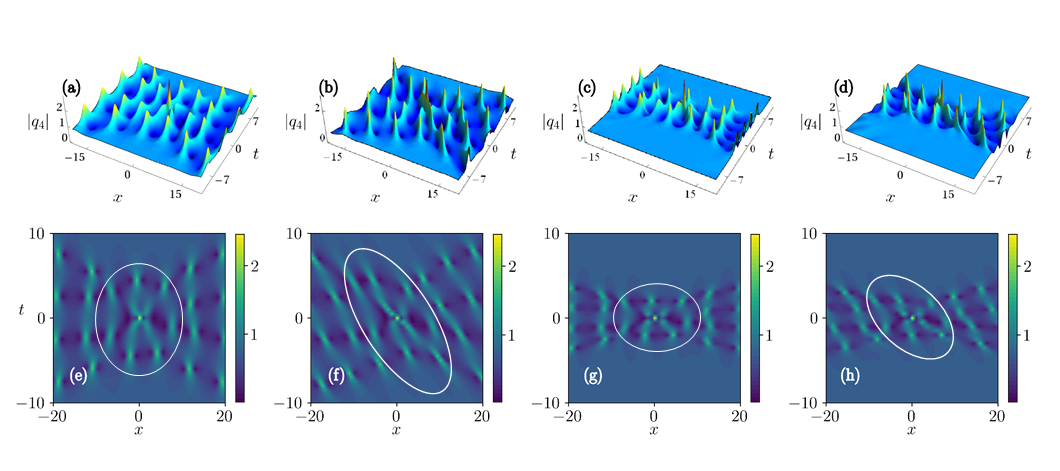}
	\caption{Fourth order circular b-p solution of with the paramter values $\lambda_1=0.5i$ and $c=0.7$ , (a) $\alpha= \gamma=0$, (b) $\alpha=0.5;\;\gamma=0$, (c) $\alpha=0;\;\gamma=0.5$, (d) $\alpha= \gamma = 0.5$. Figs.~(e)-(h) are the corresponding contour illustration of Figs.~(a)-(d).}
\end{figure}
\section{Asymptotic analysis}
In this section, we study the asymptotic behaviour of the smooth positons of ENLSE. The behaviour of two positon solution in the asymptotic regime $t\rightarrow\pm \infty$ becomes trivial since in the denominator the variable $t$ appears both in exponential and polynomial forms (see Eq.~\eqref{sp2}). Hence we adapt the procedure followed by Cen \textit{et. al.} \cite{kdv2,Hirota2} to investigate the asymptotic behaviour of second order smooth positon solution of the ENLSE. In this analysis, we compare the maximum of one soliton solution with the shifted multi positon solutions in the asymptotic limits. 
\begin{thm}
The time dependent displacement between two positon, after collision, is
\begin{subequations}
	\begin{eqnarray}
		\Delta(t)= \dfrac{1}{a}\log[4 a^2 |t| \sqrt{\hat \kappa}],\label{g}
	\end{eqnarray}
	where
	\begin{equation}
		\hat\kappa = a^2 \alpha ^2+(1+3 b \alpha )^2+4 \left(a^2-3 b^2+b \left(a^2-9 b^2\right) \alpha \right) \gamma +4
		\left(a^4-2 a^2 b^2+9 b^4\right) \gamma ^2.
	\end{equation}
\label{11}
\end{subequations}
\end{thm}
\begin{proof}
First we obtain one soliton solution of Eq.~\eqref{eq1}, by imposing the restrictions $N=1$ and $\lambda= (i a-b)/2 $, $ a, b \in \mathbb{R}$, in Eq.~\eqref{D1}. As a result, we obtain
\begin{subequations}\begin{equation}
	q_1(x,t) = a e^{iz}
	\sech \left[a (x+t (-2 b +( a^2 - 3 b^2) \alpha - 4  b \gamma(a^2 - 
	2 b^2)))\right],\label{a5}
\end{equation}
where 
\begin{equation}
	z= bx +t(a^2 -b^2 + \alpha(3 a^2-b^2)+\gamma(a^4-6a^2b^2+b^4)).\label{b}
\end{equation}
\end{subequations}
The absolute value of one soliton solution \eqref{a5} is
\begin{equation}
|q_1(x,t)|=	\frac{2 a e^{a \left(x+a^2 t \alpha -3 b^2 t \alpha +4 b^3 t \gamma -2 b \left(t+2 a^2 t \gamma \right)\right)}}{1+e^{2 a \left(x+a^2
		t \alpha -3 b^2 t \alpha +4 b^3 t \gamma -2 b \left(t+2 a^2 t \gamma \right)\right)}}.\label{a2}
\end{equation}
\par To track a distinct point on the one soliton, we choose a reference frame by fixing the wave coordinate. From \eqref{a2}, we can determine the maximum point as 
\begin{equation}
	\tilde{x}=t (3 b^2 \alpha-a^2 \alpha  - 4 b^3 \gamma + 
	b (2 + 4 a^2 \gamma)).
\end{equation}
Substituting this value back in Eq.~\eqref{a2}, we can obtain the maximum amplitude of one soliton solution \eqref{a2} as
\begin{equation}
	|q_1( \tilde{x},t)| = a.\label{a1}
\end{equation}
\par The explicit expressions of the second order smooth positon solution is given in Eq.~\eqref{sp2}. To find the asymptotic limit of this second order smooth positon solution, we replace $x \rightarrow \tilde{x}+ \Delta$, where $\Delta$ is a constant in the two positon solution. At $t\rightarrow\pm\infty$, the solution becomes zero since the variable $|t|$ appears as a polynomial in the denominator. To obtain a finite value at $t\rightarrow\pm\infty$, we introduce a logarithmic term with time dependence in the function $\Delta$, that is \cite{kdv2}
\begin{equation}
\Delta(t)= (1/a)\log(\kappa|t|),\label{a6}	
\end{equation} 
 where $\kappa$ is an undetermined constant. We then collect the dominant terms of $|t|$ in both the numerator and denominator in the absolute value of the two positon solution. The resultant action yields
 \begin{subequations}
 \begin{equation}
 	|q_2(\tilde{x}+ \Delta(t),t)|=\dfrac{A_1}{A_2},\label{a}	
 \end{equation}
where
 \begin{eqnarray}
 A_1&=& 8a^3\kappa (1 - 36 b^3 \alpha \gamma + 4 \gamma^2(a^4  + 
		9 b^4 ) + a^2 (\alpha^2 + 4 \gamma) \\&&+ 
		2 b \alpha (3 + 2 a^2 \gamma) + 
		b^2 (9 \alpha^2 - 
		4 \gamma (3 + 
		2 a^2 \gamma)))^{1/2},\notag\\
	A_2&=&64 a^8 \gamma^2 + 
		16 a^4 (1 + 3 b \alpha - 6 b^2 \gamma)^2 + 
		16 a^6 (\alpha^2 + 4 b \alpha \gamma\\&& + 
		4 \gamma (1 - 2 b^2 \gamma)) + \kappa^2.\notag
\end{eqnarray}
\end{subequations}
 \begin{figure}[!ht]
 	\centering
 	\begin{minipage}[b]{0.3\textwidth}
 		\includegraphics[width=6cm,height=6cm]{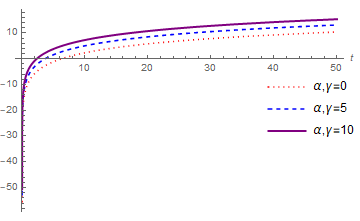}
 		\caption{Time dependent displacement $\Delta(t)$ for the values $a=0.2$ and $ b=0.5$}
 		\label{fig:deltaplot}
 	\end{minipage}
 	\hfill
 	\begin{minipage}[b]{0.5\textwidth}
 		\centering
 		\includegraphics[width=7cm,height=6.5cm]{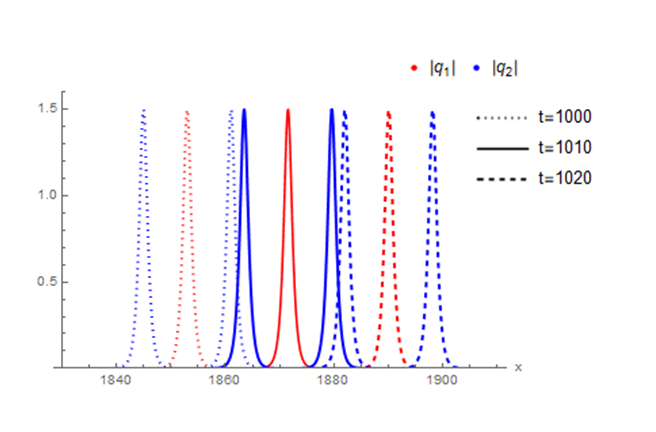}
 		\caption{Second order smooth positon solution $|q_2(x,t)|$ and one soliton solution $|q_1(x,t)|$ with the parameter values $a = 1.5, b = 1.6, \alpha = 0.3$ and $\gamma = 1.5$}
 		\label{fig:2dplot}
 	\end{minipage}
 \end{figure}
To determine the value of $\kappa$, we equate Eq.~\eqref{a} with the maximum of one soliton given in Eq.~\eqref{a1}. Doing so, we find
 \begin{eqnarray}\label{a7}
 \kappa&=& 4 a^2(a^2 \alpha ^2+(1+3 b \alpha )^2+4 \left(a^2-3 b^2+b \left(a^2-9 b^2\right) \alpha \right) \gamma\notag \\&&+4
 \left(a^4-2 a^2 b^2+9 b^4\right) \gamma ^2)^{1/2}.
 \end{eqnarray}
\par The time dependent shift tracks the stable one soliton within the absolute two positon solution. Substituting \eqref{a7} in \eqref{a6} we can obtain the time dependent displacement as given in (\ref{11}). 
\end{proof}
\begin{thm} In the asymptotic limits, the modulus of the shifted two positon solution is equal to the amplitude of one soliton solution.
\end{thm}
\begin{proof}	
Let us rewrite the two positon solution \eqref{sp2} in the form $q_2 = q_{2r}+ iq_{2i}$. The real ($q_{2r}$) and imaginary parts ($q_{2i}$) of this solution relies on the function $z(x,t)$, given in Eq.~\eqref{b}, which appears in the arguments of the sine and cosine functions. The internal oscillations are produced for different values of $z(x,t)$ which makes it unfeasible to track the constant amplitude \cite{SineG}. To obtain a constant amplitude for the two positon solution with the same overall speed, we consider only the enveloping function by fixing $z(x,t)=z$ (a constant). Implementing this, we find 
\begin{subequations} 
\begin{equation}
	q_{2r}(x,t) =\dfrac{\Omega_{12r}}{\Omega_{22r}},
\end{equation}
	with
	\begin{eqnarray}
\Omega_{12r}&=&-(4 a e^{
			a (x -\tilde{x})} ( (-1 + 2 a b t - 
		a x - 3 a^3 t \alpha + 3 a b^2 t \alpha + 
	12 a^3 b t \gamma - 4 a b^3 t \gamma\notag\\&& + 
e^{2 a (x -\tilde{x})} ( 
3 a^3 t (\alpha - 4 b \gamma) + 
a (x-2 b t  - 3 b^2 t \alpha + 
4 b^3 t \gamma)))\cos[z]\notag\\&& + 
		2 a^2 (1 +e^{2 a (x - \tilde{x})}) t (1 + 
		3 b \alpha + 2 a^2 \gamma - 6 b^2 \gamma) \sin[
		z])),\notag\\
	\Omega_{22r}&= &1 +e^{
			4 a (x -\tilde{x})} + 
		2e^{2 a (x -\tilde{x})} (1 + 
		32 a^8 t^2 \gamma^2  + 
		2 a^6 t^2 (9 \alpha^2 - \gamma(24 b \alpha\notag\\&& + 
		16(1 + 3 b^2 \gamma))) + 
		2 a^2 (x-2 b t  - 3 b^2 t \alpha + 4 b^3 t \gamma)^2+ 
		4 a^4 t (3 x (\alpha - 4 b \gamma)\notag\\&& + 
		t (2 + 6 b \alpha + 9 b^2 \alpha^2 - 
		24 b^3 \alpha \gamma + 24 b^4 \gamma^2))).		
\end{eqnarray}
\end{subequations}
\par A similar expression can also be found for the imaginary part of the two positon solution.  Since the time dependent shift has been determined, we take $x\rightarrow \tilde{x}\pm \Delta(t)$ in the two positon solution and collect the dominant terms of $|t|$. The real part of the two positon solution yields
	\begin{subequations}\label{asym}
		\begin{eqnarray}
		\lim_{t\rightarrow\pm \infty}q_{2r}(\tilde{x}+\Delta(t))&=&\mp\dfrac{a}{ \sqrt{\hat\kappa}}(a (\alpha -4 b \gamma ) \cos[z]\notag\\&&+(1+3 b \alpha +2 a^2 \gamma -6 b^2 \gamma ) \sin[z]),\\
		\lim_{t\rightarrow\pm \infty}q_{2r}(\tilde{x}-\Delta(t))&=&\pm\dfrac{a}{ \sqrt{\hat \kappa}} (a (\alpha -4 b \gamma ) \cos[z]\notag\\&&-(1+3 b \alpha +2 a^2 \gamma -6 b^2 \gamma ) \sin[z]).
	\end{eqnarray} 
For the imaginary part of the two positon solution, we find \begin{eqnarray}
 		\lim_{t\rightarrow\pm \infty}q_{2i}(\tilde{x}+\Delta(t))&=&\pm \dfrac{a}{ \sqrt{\hat \kappa}} ((1+3 b \alpha +2 a^2 \gamma -6 b^2 \gamma ) \cos[z]\notag\\&&-a (\alpha -4 b \gamma ) \sin[z]),
 		\\
 		\lim_{t\rightarrow\pm \infty}q_{2i}(\tilde{x}-\Delta(t))&=&\pm \dfrac{a}{ \sqrt{\hat \kappa}} ((1+3 b \alpha +2 a^2 \gamma -6 b^2 \gamma ) \cos[z]\notag\\&&+a (\alpha -4 b \gamma ) \sin[z]).
 \end{eqnarray} 
	\end{subequations}
The modulus of the shifted two positon solution in all the above four asymptotic limits yield the same value, that is
 \begin{equation}
 	\lim_{t\rightarrow\pm \infty}|q_2^{a,b}(\tilde{x}\pm\Delta(t))|=a.	
 \end{equation}
\end{proof}
\par From the above result, we conclude that the individual one soliton constituents of the two positon solution are  same in both the limits $t\rightarrow\pm \infty$ which can also be confirmed from Fig.~13. Therefore the two one solitons have interchanged their positions with an overall displacement $2 \Delta(t)$. If $\gamma\rightarrow0$ in Eq.~\eqref{11}, we obtain the time dependent displacement of the Hirota equation. The obtained expression exactly coincides with the one reported recently in \cite{Hirota2} for the Hirota equation.  Interestingly, in the limit $\alpha,\gamma\rightarrow0$, the time dependent displacement expression which we found above also matches with the one obtained through inverse scattering method carried out for the NLS equation \cite{nls2}. The time dependent displacement for the smooth positon solution depends on the higher order nonlinear and dispersion parameters $\alpha$ and $\gamma$. While increasing the value of $\alpha$ and $\gamma$, the time dependent shift increases which can be seen in Fig.~12. By tuning certain parameters, the asymptotic limit of Eq.~\eqref{asym} agrees exactly with the result obtained in \cite{Hirota2} for $\gamma\rightarrow 0$. If $\alpha=0$ in Eqs.~\eqref{asym} and \eqref{g}, then the resultant value admits the expression for asymptotic limit of two positon solution and time dependent displacement of the fourth order NLS equation respectively.  It will be interesting to investigate the conserved quantities associated with positons.  It has been shown that the energy of two positon solution will be twice of energy of one positon, that is  $E(q[2])=2E(q[1])$.  In general, for N-positon solutions, $E(q[N])=NE(q[1])$.  We plan to investigate this aspect numerically in near future.

\section{Conclusion}
\par In this work, we have derived higher order degenerate soliton solutions for the ENLSE \eqref{eq1}. Using GDT method, we have constructed smooth positons and b-p solutions of various orders (second, third and fourth).  We have demonstrated that from the constructed solution one can deduce the smooth positon and b-p solutions of NLS, Hiorta and a fourth order NLS equation. We have analyzed the effect of higher order odd and even nonlinear terms on the basic smooth positon and b-p solutions.  Our investigations reveal that the higher order nonlinear terms impose greater compression effect and these two parameters also  tilt the waves. The central region of the b-p solution exhibits a similar structure as that of rogue waves. By introducing certain parameters in the b-p solution, we have explored triangular and circular pattern of waves in the central region of the b-p solutions. In addition to the above, we have studied the behaviour of two positon solution in the asymptotic limits and demonstrated that they exhibit time dependent phase shift during collision. We have also derived the expression for time dependent displacement. Since positon solutions do not exhibit energy exchange during collision they prevent data loss in nonlinear optical fiber systems.  Moreover, multi-positons travel like a single component and therefore they can  model the tidal bore phenomenon in which equal amplitude waves travel simultaneously for several kilometers. Hence, the results which we have presented in this paper will be helpful for both the optics and water wave research community.

\section*{Acknowledgments}

 SM thanks MoE-RUSA 2.0 Physical Sciences, Government of India for providing a fellowship to carry out this work. NVP wishes to thank Department of Science and Technology (DST), India for the financial support under Women Scientist Scheme-A. The work of MS forms part of a research project sponsored by NBHM, Government of India, under the Grant No. 02011/20/2018 NBHM (R.P)/R\&D II/15064. SR acknowledges the financial support from MoE-RUSA 2.0 Physical Sciences, Government of India.

 \bibliographystyle{elsarticle-num}

\end{document}